\documentclass[a4paper,english,numberwithinsect]{2019_lipics}

\title{The Tree Stabbing Number is not Monotone} 

\author{Wolfgang Mulzer}{Freie Universit\"at Berlin, Germany }{mulzer@inf.fu-berlin.de}{}{}

\author{Johannes Obenaus}{Freie Universit\"at Berlin, Germany}{johannes.obenaus@fu-berlin.de}{}{Partially supported by ERC StG 757609}

\authorrunning{W. Mulzer and J. Obenaus}

\Copyright{Wolfgang Mulzer and Johannes Obenaus}



\category{}

\relatedversion{}

\supplement{}

\nolinenumbers 

\hideLIPIcs  

\newcommand{\R}{\mathbb{R}}

\newcommand{\treestab}[1]{\textsc{tree-stab}(#1)}
\newcommand{\tristab}[1]{\textsc{tri-stab}(#1)}
\newcommand{\pathstab}[1]{\textsc{path-stab}(#1)}
\newcommand{\matstab}[1]{\textsc{mat-stab}(#1)}

\newcommand{\placeholder}{$\cdot$}

\DeclareMathOperator{\conv}{conv}

\usepackage{csquotes}
\usepackage{enumitem}
\usepackage{subcaption}
\usepackage{cite}

\usepackage{chngcntr}
\usepackage{apptools}
\AtAppendix{\counterwithin{theorem}{section}}

\begin{document}

\maketitle

\begin{abstract}
Let $P \subseteq \R^2$ be a set of points and $T$ be a spanning tree of $P$. The \emph{stabbing number} of $T$ is the maximum number of intersections any line in the plane determines with the edges of $T$. The \emph{tree stabbing number} of $P$ is the minimum stabbing number of any spanning tree of $P$. We prove that the tree stabbing number is not a monotone parameter, i.e.,\ there exist point sets $P \subsetneq P'$ such that \treestab{$P$}~$>$~\treestab{$P'$}, answering a question by Eppstein \cite[Open~Problem~17.5]{eppstein_2018}.
\end{abstract}

\section{Introduction}

Let $P\subseteq \R^2$ be a set of points in general position, i.e., no three points lie on a common line. A \emph{geometric graph} $G = (P,E)$ is a graph equipped with a drawing where edges are realized as straight-line segments. The \emph{stabbing number} of $G$ is the maximum number of proper intersections that any line in the plane determines with the edges of $G$. Let $\mathcal{G}$ be a graph class (e.g., trees, paths, triangulations, perfect matchings etc.). The \emph{$\mathcal{G}$-stabbing number} of $P$ is the minimum stabbing number of any geometric graph $G = (P, E)$ belonging to $\mathcal{G}$ (as a function of $P$).

Stabbing numbers are a classic topic in computational geometry and received a lot of attention both from an algorithmic as well as from a combinatorial perspective. We mainly focus on the stabbing number of spanning trees (see, e.g., \cite{Welzl1992} for more information), which has numerous applications. For instance, Welzl \cite{Welzl1988PartitionTF} used spanning trees with low stabbing number to efficiently answer triangle range searching queries, Agarwal \cite{Agarwal:1992:RSO:Ray_shooting} used them in the context of ray shooting (also see \cite{Edelsbrunner1989, AGARWAL1993229} for more examples). Furthermore, Fekete, L\"ubbecke and Meijer \cite{Fekete2008} proved $\mathcal{NP}$-hardness of stabbing numbers for several graph classes, namely for spanning trees, triangulations and matchings, though for paths this question remains open.

It is natural to ask whether stabbing numbers are monotone, i.e., does it hold for any pointset $P \subseteq \R^2$ that the $\mathcal{G}$-stabbing number of $P$ is not smaller than the $\mathcal{G}$-stabbing number of any proper subset $P' \subsetneq P$. Recently, Eppstein \cite{eppstein_2018} gave a detailed analysis of several parameters that are monotone and depend only on the point set's order type. Clearly, stabbing numbers depend only on the order type. Eppstein observed that the path stabbing number is monotone \cite[Observation~17.4]{eppstein_2018} and asked whether this is also the case for the tree stabbing number \cite[Open~Problem~17.5]{eppstein_2018}. We prove that neither the tree stabbing number (Corollary~\ref{cor:treestab}) nor the triangulation stabbing number (Corollary~\ref{cor:tristab}) nor the matching stabbing number (Corollary~\ref{cor:matstab}) are monotone. A more detailed analysis can also be found in the second author's Master thesis~\cite{mastersthesis}. Each of the following sections is dedicated to one graph class.

\section{Path Stabbing Number} \label{sec:pathstab}

For completeness we repeat the main argument that the path stabbing number, denoted by \pathstab{\placeholder}, is monotone, which can be found in \cite[Observation 17.4]{eppstein_2018} for example.

\begin{lemma} \label{lem:path}
  Let $G$ be a geometric graph. The following two operations do not increase the stabbing number of $G$:
\begin{enumerate}
    \item Removing a vertex of degree 1.
    \item Replacing a vertex $v$ of degree 2 with the segment connecting its two neighbours $w_1, w_2$. 
\end{enumerate}
\end{lemma}

\begin{proof}
Clearly, the first operation cannot increase the stabbing number, since it does not add any new segments.

For the second part, let $G'$ be the geometric graph obtained from $G$ by performing operation 2 and let $\ell$ be an arbitrary line. If $\ell$ has strictly less than \textsc{Stabbing-Number}$(G)$ intersections in $G$, it has at most \textsc{Stabbing-Number}$(G)$ intersections in $G'$, since we added only one segment. Otherwise, if $\ell$ has \textsc{Stabbing-Number}$(G)$ intersections in $G$, it clearly does not pass through any vertex of $G$ and if $\ell$ intersects the newly inserted segment $\overline{w_1w_2}$ it must have also intersected either $\overline{w_1v}$ or $\overline{vw_2}$. 
\end{proof}

\begin{corollary}
\label{cor:path_stab}
\pathstab{\placeholder} is monotone.
\end{corollary}

\section{Tree Stabbing Number} \label{sec:treestab}

We construct point sets $P_1 \subsetneq P_2$ of size $n$ and $n+1$ such that \treestab{$P_1$} > \treestab{$P_2$}. The point $p \in P_2\setminus P_1$ we want to remove, must, of course, have degree at least 3 in any spanning tree of minimum stabbing of $P_2$, since otherwise the arguments of Lemma \ref{lem:path} apply. 

Our construction, which is depicted in Figure~\ref{fig:treestab_counter} (a), is as follows. Start with a unit circle around the origin $O$ and place 3 evenly distributed points $x_1,x_2,x_3$ on this circle (in counterclockwise order). Next, add an \enquote{arm} consisting of 2 points $y_i, z_i$ ($i=1,2,3$) at each of the $x_i$ (outside the circle) such that the points $O,x_i,y_i,z_i$ form a convex chain for $i=1,2,3$ (which are all three oriented the same way). These arms need to be flat enough, i.e., the line supporting the segment $\overline{x_iy_i}$ must intersect the interior of the segment $\overline{Ox_{i+2}}$ (indices are taken modulo 3), but also curved enough, i.e., the line supporting the segment $\overline{y_iz_i}$ must have the remaining 8 points on the same side. In particular, there are lines intersecting the segments $\overline{x_iy_i}$, $\overline{y_iz_i}$ and also $\overline{Ox_{i+2}}$ on the one hand and $\overline{y_{i+2}z_{i+2}}$ on the other hand (the red lines in Figure~\ref{fig:treestab_counter} (a)). If there is no danger of confusion, we might omit that indices are taken modulo 3 (as in the previous sentence).

\begin{figure}
    \centering
    \includegraphics[width=\textwidth, page=6]{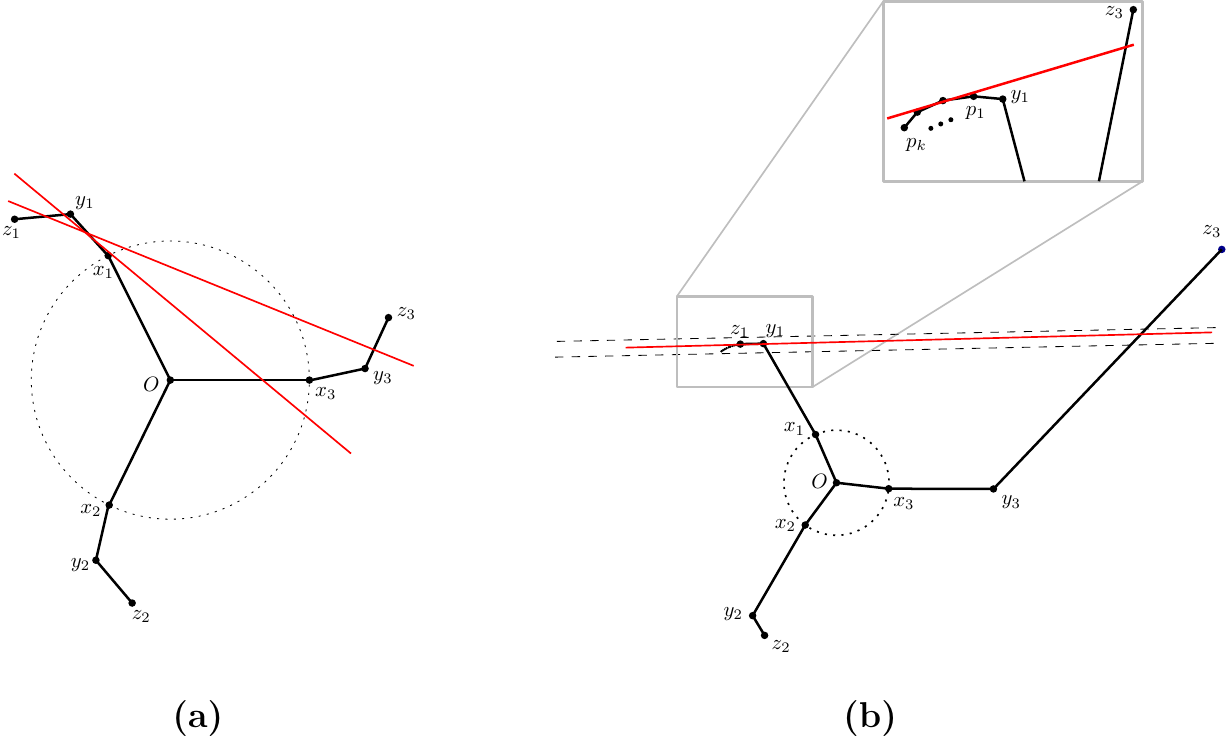}
    \caption{Illustration of a set of (a) 9 points and (b) n points such that removing the point $O$ increases the tree stabbing number.}
    \label{fig:treestab_counter}
\end{figure}

Define the two point sets $P_1, P_2$ (which are both in general position) to be 
\[
P_1 = \{ x_1, y_1, z_1, x_2, y_2, z_2, x_3, y_3, z_3 \}, \quad \qquad P_2 = P_1 \cup \{ O \}.
\]

\begin{lemma} \label{lem:treestabbing_special}
It holds that $\treestab{P_1} = 4$ and $\treestab{P_2} \leq 3$.
\end{lemma}

\begin{proof}
This result was obtained by a computer-aided brute-force search (the source code is available on github \cite{code}). In order to compute the stabbing number of a given geometric graph spanning some point set, it is enough to consider a \emph{representative set} $H_P$ of lines. For any line $\ell$ that partitions the point set into two non-empty subsets, there is a line in the representative set inducing the same partitioning. For an $n$-point set in general position, the size of a representative set is $\binom{n}{2}$ (see appendix, Lemma~\ref{lem:repr_lines}). Hence, we have $|H_{P_1}| = 36$ and $|H_{P_2}| = 45$. The sets $H_{P_1}$ and $H_{P_2}$ were also obtained by computer assistance. Any pair of points induces four distinct representative lines, computing these and removing duplicates yields $H_{P_1}$ and $H_{P_2}$ (as in \cite{mulzer_eurocg_2014} for example).

Now, it is enough to compute -- for all $9^7 = 4782969$ possible spanning trees on $P_1$ -- their intersections with the lines in $H_{P_1}$, yielding \treestab{$P_1$} = 4.

On the other hand, for $P_2$ the spanning tree depicted in Figure~\ref{fig:treestab_counter} has stabbing number 3 (again by computing all intersections with lines in $H_{P_2}$) implying \treestab{$P_2$}$ \leq 3$.
\end{proof}

Next, we generalize this construction to arbitrarily large point sets. We simply replace one of the $z_i$ (say $z_1$) by a convex chain $C$ consisting of $k$ points $p_1, \dots, p_k$ (see Figure~\ref{fig:treestab_counter} (b)). Denote the convex chains $x_1y_1C$, $x_2y_2z_2$ and $x_3y_3z_3$ by $C_1$, $C_2$ and $C_3$.

Our goal will be to remove all but two points of $C \cup \{y_1\}$ to get back to our 9-point setting. Of course, it is crucial to keep the relative position of the points as it is in the 9-point set. Thus, place the points $p_1, \dots, p_k$ such that: 

\begin{enumerate}
    \item $O,x_1,y_1,p_1,\dots,p_k$ forms a convex chain.
    \item close enough to $y_1$, so that the order type of the resulting point set is the same no matter which $k-1$ of the points in $C \cup \{y_1\}$ we remove. In particular, no line through any two points not belonging to $y_1,p_1,\dots,p_k$ may separate these points.
    \item for any two segments formed by any triple of points in $C_1$ (consecutively along the convex chain) there is a line intersecting these two segments and also $\overline{y_3z_3}$. To achieve this, $C$ needs to be sufficiently flat and $z_3$ needs to be pushed further away.
\end{enumerate}

Note that Lemma~\ref{lem:treestabbing_special} has been verified to still hold after the modification of pushing $z_3$ further out. Before proving that this construction fulfills the desired properties, we need one more preliminary lemma (see Figure \ref{fig:contraction}).

\begin{lemma} \label{lem:contraction}
Let $G = (V,E)$ be a forest with $c$ connected components and $|V| \geq 4$. Mark three of the vertices as special (call them $v_1, v_2, v_3$) and iteratively remove/replace vertices of degree 1 and 2 (as in Lemma~\ref{lem:path}) until no non-special vertex of degree $\leq 2$ remains. Then the resulting graph is a forest and consists of the three special vertices and at most one non-special vertex.
\end{lemma}

\begin{figure}
    \centering
    \includegraphics[width=\textwidth]{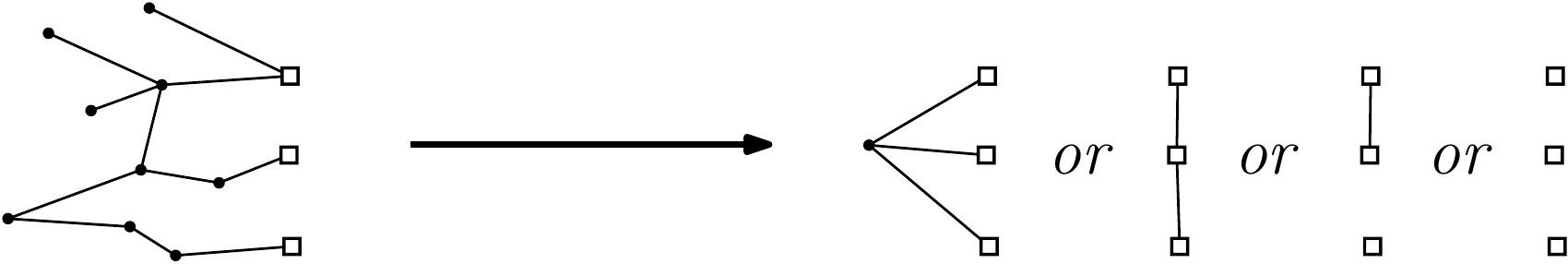}
    \caption{Illustration of Lemma~\ref{lem:contraction}. Special vertices are depicted as squares. Other vertices of degree 1 or 2 are successively removed.}
    \label{fig:contraction}
\end{figure}

\begin{proof}

Let $G' = (V',E')$ be the graph that was obtained from $G$ by repeatedly performing the two operations of Lemma~\ref{lem:contraction} and let $n'$ denote its number of vertices (including the three special). Clearly $G'$ is a forest, since both operations decrease the number of vertices and the number of edges by exactly 1 and cannot create cycles.

Furthermore, all non-special vertices have degree at least three. Then -- using the handshaking lemma and the fact that the forest $G'$ has $n'-c'$ edges, where $c'$ denotes the number of connected components in $G'$ -- we obtain:
\begin{equation} \label{eq:handshake}
2|E'| = \sum_{i=1}^{n'} \deg(v_i) \geq 3(n'-3) + \sum_{i=1}^{3} \deg(v_i) \, .
\end{equation}

Observe that $c' \leq 3$ holds, since all connected components not containing a special vertex are completely removed, which follows inductively from the fact that any tree has a leaf. Therefore, it suffices to consider the following three cases:

\begin{description}
    \item[Case 1:] $c' = 1$. This implies $\sum_{i=1}^{3} \deg(v_i) \geq 3$, and hence (using Equation~\ref{eq:handshake}): 
    \[ 2(n' - 1) \geq 3n' - 6. \]
    \item[Case 2:] $c' = 2$. This implies $\sum_{i=1}^{3} \deg(v_i) \geq 2$, and hence (using Equation~\ref{eq:handshake}): 
    \[
    2(n' - 2) \geq 3n' - 7.
    \]
    
    \item[Case 3:] $c' = 3$. This implies $\sum_{i=1}^{3} \deg(v_i) \geq 0$, and hence (using Equation~\ref{eq:handshake}):
    \[
    2(n' - 3) \geq 3n' - 9.
    \]
\end{description}
The inequality in case 1 is only satisfied for $n' \leq 4$ and in case 2 and 3 only for $n' \leq 3$.

\end{proof}

Now, we are prepared to prove our main lemma.

\begin{lemma} \label{lem:tree_stabbing_general}
For any integer $n \geq 9$, there exist (planar) point sets $P'_1 \subsetneq P'_2$ of size $|P'_1| = n$ and $|P'_2| = n+1$ such that \treestab{$P'_1$} $>$ \treestab{$P'_2$}.
\end{lemma} 

\begin{proof}
Let $k = n - 8$ and define $P'_1$ and $P'_2$ as above (Figure \ref{fig:treestab_counter} (b)), replacing $z_1$ by $p_1, \dots, p_k$:
\[
P'_1 = \{ x_1, y_1, p_1, \dots, p_k, x_2, y_2, z_2, x_3, y_3, z_3 \}, \quad \qquad P'_2 = P'_1 \cup \{ O \}.
\]

On the one hand, it is straightforward to see that the spanning tree depicted in Figure~\ref{fig:treestab_counter}~(b) has stabbing number 3 (see Figure~\ref{fig:case_distinction} for an illustration) and hence \treestab{$P'_2$} $\leq 3$.

On the other hand, we show \treestab{$P'_1$} $\geq 4$ next. Assume for the sake of contradiction that there is a spanning tree $T$ of $P'_1$ with stabbing number at most 3. Our goal will be to carefully remove points from $P_1$ such that the stabbing number of $T$ cannot increase until there are only 9 points left in exactly the same relative position as in Lemma~\ref{lem:treestabbing_special}. Clearly, this would be a contradiction.

\begin{figure}
    \centering
    \includegraphics[width=\textwidth, page=5]{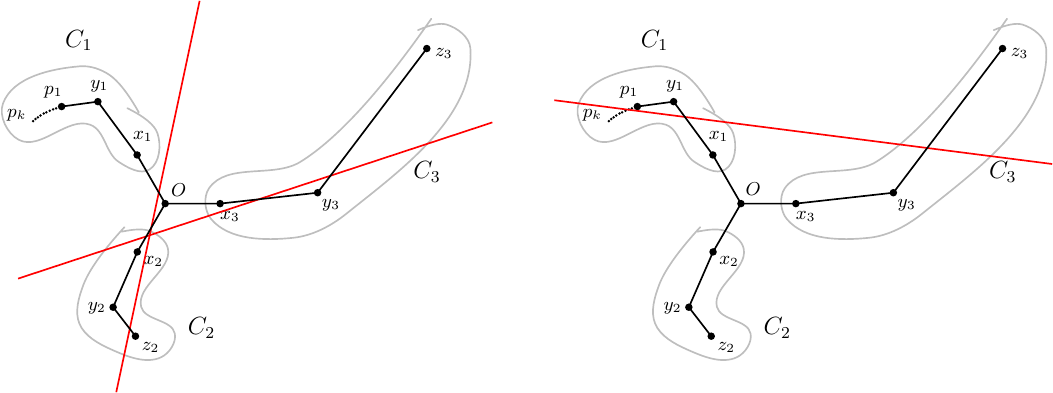}
    \caption{There is no line that intersects more than 3 segments in this spanning tree.}
    \label{fig:case_distinction}
\end{figure}

Consider the set of edges of $T$ with at least one endpoint among the points in $C_1$. There are at most 3 edges having only one endpoint in $C_1$ (we call them \emph{bridges}). If there would be more than 3 bridges, there is a line that intersects at least 4 line segments, namely a line that separates $C_1$ from the rest. Because of the same reason, not all three bridges can go to the same other component ($C_2$ or $C_3$). 

There are at most 3 points in $C_1$ that are incident to a bridge and if they are distinct, one of them needs to be $x_1$, otherwise the line separating $x_1$ from the rest of $C_1$ has 4 intersections. Pick three vertices $v_1, v_2, v_3$ in $C_1$ such that $x_1$ and any point incident to a bridge is among them and mark them as special.

Next, we apply Lemma~\ref{lem:contraction} to the subforest induced by $C_1$:

\subparagraph*{Case 1:} No non-special vertex in $C_1$ survives the removal process.
    
Then 9 points with the same order type as in Lemma~\ref{lem:treestabbing_special} and a spanning tree with stabbing number 3 remain, which is a contradiction to Lemma~\ref{lem:treestabbing_special}.

\subparagraph*{Case 2:} One non-special vertex $v$ in $C_1$ survives the removal process.
    
Then $v$ is incident to all special vertices $v_1, v_2, v_3$. If $v$ is the last vertex along $C_1$, there is obviously a line having more than three intersections. Otherwise, by construction, there is a line $\ell$ that separates $v$ from $v_1, v_2, v_3$ and at the same time $z_3$ from the rest of the point set (see Figure~\ref{fig:monotone_contraction}). In particular, $\ell$ has only $z_3$ and $v$ on one side and all other points on the other. $z_3$ cannot be adjacent to $v$, since $v$ is not incident to a bridge and therefore contributes another intersection to $\ell$. This is a contradiction to the assumption that $T$ was a spanning tree of stabbing number 3.
\end{proof}

\begin{figure}
    \centering
    \includegraphics[width=\textwidth, page=3]{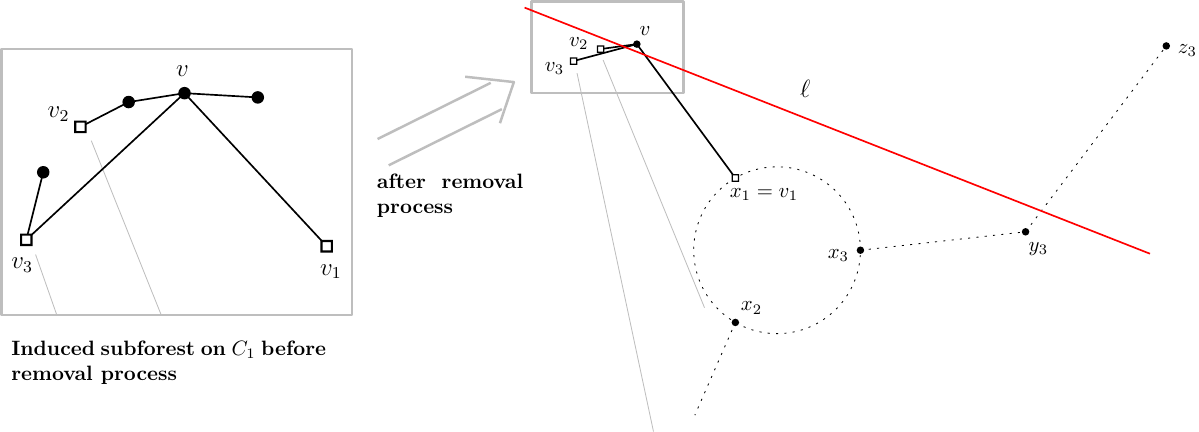}
    \caption{Illustration of Case 2. If a non-special vertex $v$ survives the removal process, the red line has too many intersections.}
    \label{fig:monotone_contraction}
\end{figure}

\begin{corollary} \label{cor:treestab}
\treestab{\placeholder} is not monotone.
\end{corollary}

\section{Triangulation Stabbing Number} \label{sec:tristab}

\begin{figure}
    \centering
    \includegraphics[width=0.75\textwidth, page=1]{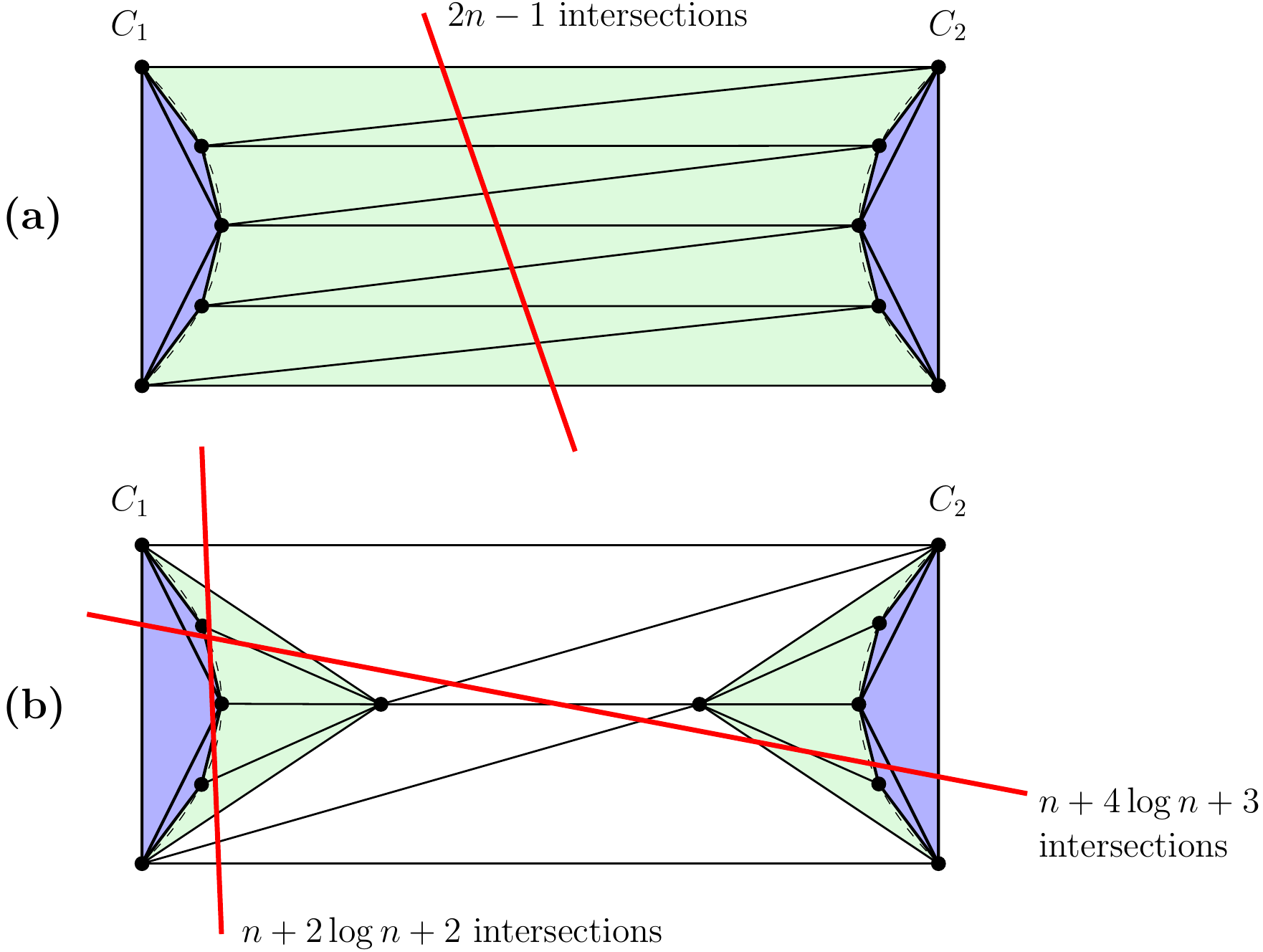}
    \caption{Two symmetric chains in (a) might have a larger triangulation stabbing number compared to the same point set with additional points inbetween (b).}
    \label{fig:tri_stab3}
\end{figure}

We denote the triangulation stabbing number by \tristab{\placeholder}. Proving non-monotonicity of \tristab{\placeholder} is much simpler, only exploiting the additional structure enforced by triangulations. Consider two symmetric convex chains $C_1 = \{p_1, \ldots, p_n\}$ and $C_2 = \{p_1', \ldots, p_n'\}$ (sufficiently flat) each consisting of $n$ points and facing each other as depicted in Figure~\ref{fig:tri_stab3}~(a). These points constitute the point set $P$. $P'$ consists of the same $2n$ points and two more (slightly perturbed) points added on the line segment connecting the two middle points of $C_1$ and $C_2$ (as in Figure~\ref{fig:tri_stab3}~(b)). Then the following holds:

\begin{lemma} \label{lem:tri_stab}
\tristab{P} $\geq 2n - 1$ and \tristab{P'} $\leq n + 4\log n + 3$.
\end{lemma}

\begin{proof}

Any triangulation of $P$ must have $2n-1$ segments connecting a point from $C_1$ with a point from $C_2$ (the green area in Figure~\ref{fig:tri_stab3} (a)). Hence, 
\[
\tristab{P} \geq 2n - 1.
\]
On the other hand, the triangulation of $P'$ depicted in Figure~\ref{fig:tri_stab3} (b) has stabbing number $n + 4\log n + 3$, which can be seen as follows. The two green areas contain $n$ segments each and are constructed in such a way that any line $\ell$ may intersect at most $n$ segments from both green areas. For this, the two points $p_{n/2 - 1}$ and $p_{n/2 + 1}$ need to be sufficiently far from $p_{n/2}$. Figure~\ref{fig:tri_stab4} illustrates the area that contains all lines which intersect line segments of \enquote{upper} ($p_1, \ldots, p_{n/2}$) and \enquote{lower} ($p_{n/2}, \ldots, p_{n}$) half of the convex chain in the green region.

\begin{figure}
    \centering
    \includegraphics[width=0.65\textwidth, page=2]{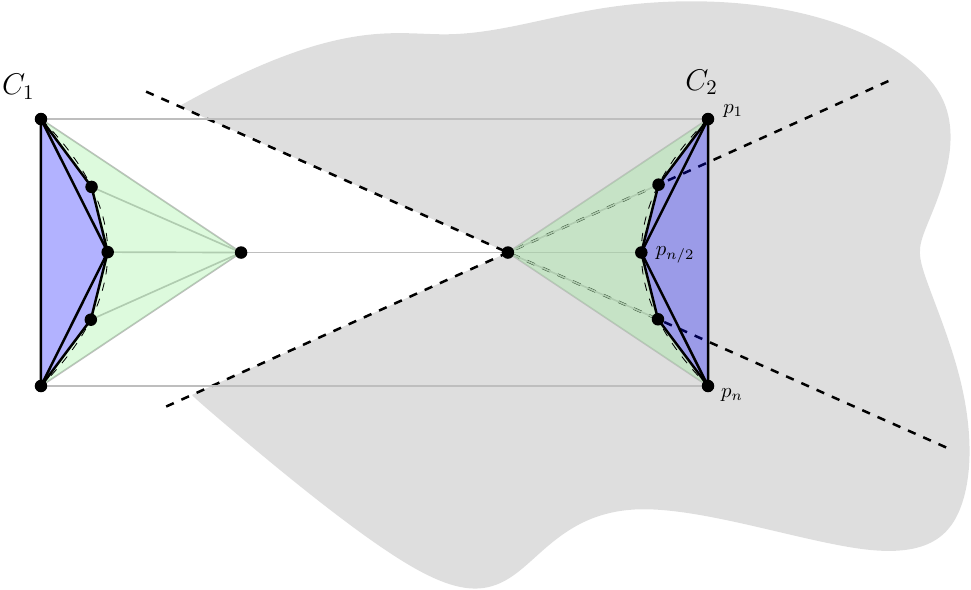}
    \caption{Any line that simultaneously intersects a segment connected to the \enquote{upper} half of the convex chain $C_2$ and a segment connected to the \enquote{lower} half of $C_2$ (both in the green region) must be fully contained in the shaded area and hence, cannot intersect any green segment of $C_1$ (and also the other way around).}
    \label{fig:tri_stab4}
\end{figure}

Since $C_1$ and $C_2$ are convex, it is clear that $\ell$ may accumulate $4\log n$ more intersections in the blue areas. The white area contains only a constant number of segments, in total $\ell$ has at most $n + 4\log n + 3$ intersections. Hence,
\[
\tristab{P'} \leq n + 4\log n + 3.
\]
\end{proof}

\begin{corollary} \label{cor:tristab}
\tristab{\placeholder} is not monotone.
\end{corollary}

\section{Matching Stabbing Number} \label{sec:matstab}

First note that the point sets in the case of matchings have to be of even size and all matchings are perfect.

Take $k$ points $p_1, \dots, p_k$ in convex position and one point $x$ inside such that any segment $\overline{xp_i}$ is intersected by some $\overline{p_jp_k}$. Next, double all points within a small enough $\varepsilon$-radius (preserving general position) and for a point $p$ name the partner point $p'$ (see Figure~\ref{fig:mat_stab}). 

Define the point sets $P_1$ and $P_2$ to be:
\[
P_2 = \{ x, x', p_1, \dots, p_k, p'_1, \dots, p'_k, \}, \quad \qquad
P_1 = P_2 \setminus \{ x', p'_1 \}.
\]

\begin{lemma} \label{lem:mat_stab}
It holds that \matstab{$P_1$} $\geq 3$ and \matstab{$P_2$} $\leq 2$.
\end{lemma}

\begin{proof}
Clearly, the perfect matching that assigns an edge to all partner points has stabbing number 2 and hence \matstab{$P_2$} $\leq 2$.

On the other hand, we show \matstab{$P_1$} $\geq 3$, which can be seen as follows. Let $M$ be a perfect matching in $P_1$ and consider the point $p_i$ (or $p'_i$) that is connected to $x$. Also consider the points $p_j$ and $p_k$ such that $\overline{xp_i}$ and $\overline{p_jp_k}$ intersect. At least one of the points $p'_j$ and $p'_k$ is present in $P_1$. 

\begin{description}
    \item[Case 1:] If $\{p_j, p'_j \}$ and $\{p_k, p'_k \}$ are both part of $M$, there is a line intersecting the three segments $\overline{xp_i}$, $\overline{p_jp'_j}$ and $\overline{p_kp'_k}$. 
    
    \item[Case 2:] If $\{p_j, p'_j \}$ and $\{p_k, p'_k \}$ are both not part of $M$, then one of the four lines depicted in Figure \ref{fig:mat_stab}~(b) has three intersections.
    
    \item[Case 3:] If exactly one of the edges $\{p_j, p'_j \}$ or $\{p_k, p'_k \}$ is part of $M$, then one of the two lines depicted in Figure \ref{fig:mat_stab}~(c) has three intersections.
\end{description}
\end{proof}

\begin{corollary} \label{cor:matstab}
The matching stabbing number, \matstab{\placeholder}, is not monotone.
\end{corollary}

\begin{figure}
    \centering
    \includegraphics[width=0.9\textwidth]{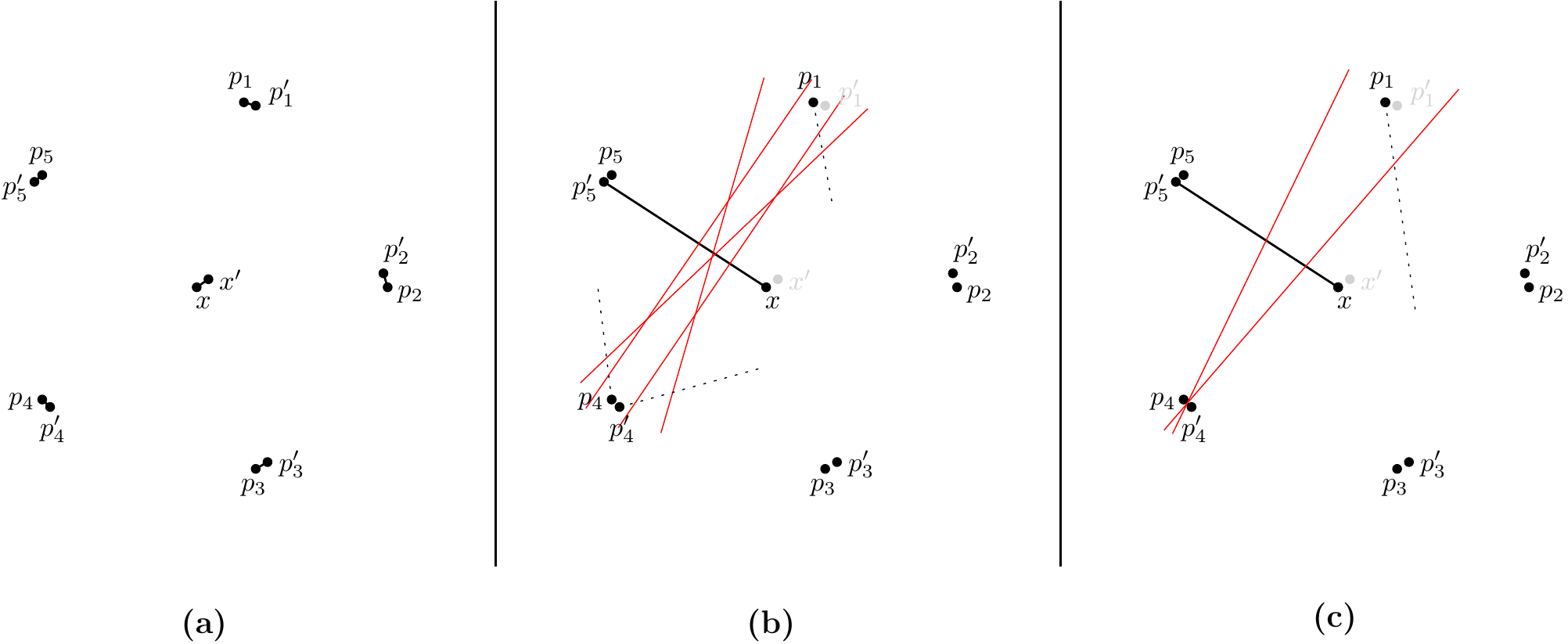}
    \caption{A point set with matching stabbing number 2 in (a) and removing $p_1$ and $x'$ results in a point set with larger matching stabbing number, illustrated in (b) and (c).}
    \label{fig:mat_stab}
\end{figure}

\section{Conclusion}

Our proof of Lemma~\ref{lem:treestabbing_special} relies on computer assistance and of course it would be interesting to turn this into a pen-and-paper proof. 

Furthermore, it is easy to generalize stabbing numbers to the context of range spaces $(X, \mathcal{R})$, where $X$ is a set and $\mathcal{R}$ a set of subsets of $X$, called \emph{ranges}. A spanning path then corresponds to a permutation of $X$ and a set $A \subseteq X$ is \emph{stabbed} by a range $r \in \mathcal{R}$ if there are $x, y \in A$ such that $x \in r$ and $y \notin r$. It is straightforward to prove Corollary \ref{cor:path_stab} in this context, but we don't know how to apply this for other graph classes. 

\bibliography{eurocg20_example}

\appendix
\section{Appendix}

\begin{lemma} \label{lem:repr_lines}
For any set $P \subseteq \R^2$ of $n$ points in general position, a set of representative lines contains exactly $\binom{n}{2}$ lines.
\end{lemma}

\begin{proof}

Let $\mathcal{P}$ be the set of \emph{realizable partitions} of $P$. Two subsets $P_1, P_2\subsetneq P$ form a \emph{realizable partition} of $P$ if $P_1 \cup P_2 = P$ and $\conv(P_1) \cap \conv(P_2) = \emptyset$, i.e., there is a line separating the convex hulls.

Define the function $f : \{ (p, q) : p, q \in P \} \to \mathcal{P}$, which takes an ordered pair of points from $P$ as input and returns a \emph{realizable partition} $(P_1, P_2)$ of $P$ as follows. For an ordered pair $(p,q)$ consider the directed line $\ell_{pq}$ through the two points (directed from $p$ to $q$). Let 
\begin{align*}
    P_l &= \{ s \in P : s \text{ is left of } \overline{pq} \} \cup p \\
    P_r &= \{ s \in P : s \text{ is right of } \overline{pq} \} \cup q
\end{align*}

and define $f((p,q)) = (P_l, P_r)$.

Since $P$ is in general position, $(P_l, P_r)$ forms a realizable partition (the line $\overline{pq}$ rotated infinitesimally in counterclockwise order separates $P_l$ and $P_r$). Also note that the line $\overline{pq}$ must be tangent to $\conv(P_1)$ and $\conv(P_2)$ and both convex hulls are contained on different sides of $\overline{pq}$.

We will show that for any realizable partition there exist exactly two pairs $(p,q)$ and $(p',q')$ which are mapped to this partition. Since there are $2 \binom{n}{2}$ ordered pairs on $P$, this proves the lemma. 

\begin{figure}
    \centering
    \begin{subfigure}[b]{0.4\textwidth}
        \includegraphics[width=\textwidth,page=1]{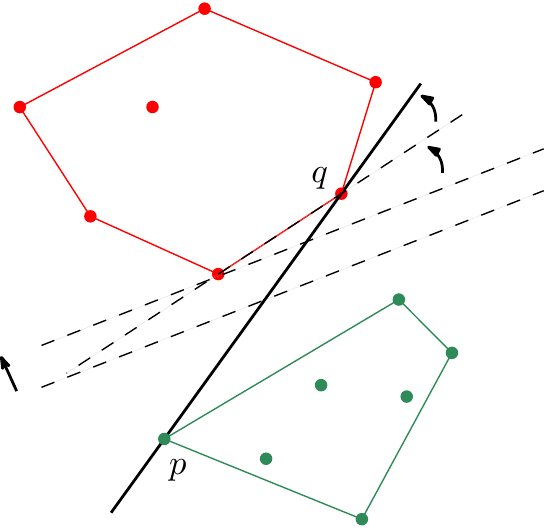}
        \caption{Rotation process to find a tuple $(p,q)$ that is mapped to the given partition.}
    \end{subfigure}
    \hfill
    \begin{subfigure}[b]{0.4\textwidth}
        \includegraphics[width=\textwidth,page=2]{rep_lines.pdf}
        \caption{Result of both rotations (clockwise and counterclockwise).}
    \end{subfigure}
    \caption{Illustration of Lemma~\ref{lem:repr_lines}. The red and green sets form a realizable partition.}
    \label{fig:rep_lines}
\end{figure}

Let $(P_1, P_2) \in \mathcal{P}$ be a realizable partition. Two pairs of points that are mapped to this partition can be found as follows. Consider a line $\ell$ separating $P_1$ and $P_2$ and shift $\ell$ towards $P_1$ until it hits a vertex. Next, rotate $\ell$ (once in clockwise and once in counterclockwise order) until it hits a vertex of the other set. Note that if we a hit a vertex of the same set, we just continue to rotate around this one (see Figure~\ref{fig:rep_lines}~(a)). This way we get exactly two distinct pairs of points (because of general position) that are mapped to the partition $(P_1, P_2)$ if ordered accordingly (see Figure~\ref{fig:rep_lines}~(b)). 

On the other hand, there is no other line tangent to $P_1$ and $P_2$ simultaneously and having both sets entirely contained on different sides. 
\end{proof}

\end{document}